\documentclass[11pt]{article}
\usepackage{fullpage}
\usepackage{hyperref}
\usepackage{palatino}  
\usepackage{mathpazo}
\usepackage{amssymb,amsmath,amsthm}

\newtheorem{theorem}{Theorem}[section]
\newtheorem{proposition}[theorem]{Proposition}

\newtheorem{claim}[theorem]{Claim}
\newtheorem{lemma}[theorem]{Lemma}

\newtheorem{corollary}[theorem]{Corollary}

\newcommand{\qedsymb}{\hfill{\rule{2mm}{2mm}}}
\renewenvironment{proof}[1][]{\begin{trivlist}
\item[\hspace{\labelsep}{\bf\noindent Proof#1:\/}] }{\qedsymb\end{trivlist}}

\def\calG{{\cal G}}
\def\calF{{\cal F}}

\def\R{\mathbb{R}}

\DeclareMathOperator{\supp}{supp}
\def\myand{\mbox{AND}}

\newcommand\Prob[2]{{\Pr_{#1}\left[ {#2} \right]}}

\newcommand\expectation{\mathop{{\mathbb{E}}}}

\newcommand{\eps}{\epsilon}
\renewcommand{\epsilon}{\varepsilon}

\DeclareMathOperator{\dist}{dist}

\newcommand{\sign}{\mathop{\mathrm{sign}}}

\newcommand{\Fset}{\mathbb{F}}         

\begin{document}

\title{{\bf The List-Decoding Size of Fourier-Sparse Boolean Functions}}

\author{
Ishay Haviv\thanks{School of Computer Science, The Academic College of Tel Aviv-Yaffo, Tel Aviv 61083, Israel.}
\and Oded Regev\thanks{Courant Institute of Mathematical Sciences, New York University. Supported by the Simons Collaboration on Algorithms and Geometry and by the National Science Foundation (NSF) under Grant No.~CCF-1320188. Any opinions, findings, and conclusions or recommendations expressed in this material are those of the authors and do not necessarily reflect the views of the NSF.}
}

\date{}

\maketitle
\thispagestyle{empty}


\begin{abstract}
A function defined on the Boolean hypercube is {\em $k$-Fourier-sparse} if it has at most $k$
nonzero Fourier coefficients.
For a function $f: \Fset_2^n \rightarrow \R$ and parameters $k$ and $d$, we prove a
strong upper bound on the number of $k$-Fourier-sparse Boolean functions that
disagree with $f$ on at most $d$ inputs. Our bound implies that the number of uniform and
independent random samples needed
for learning the class of $k$-Fourier-sparse Boolean functions on $n$ variables exactly is at
most $O(n \cdot k \log k)$.

As an application, we prove an upper bound on the query complexity of testing Booleanity of
Fourier-sparse functions. Our bound is tight up to a logarithmic factor and quadratically improves
on a result due to Gur and Tamuz (Chicago~J.~Theor.~Comput.~Sci.,~2013).
\end{abstract}


\section{Introduction}

Functions defined on the Boolean hypercube $\{0,1\}^n = \Fset_2^n$
are fundamental objects in theoretical computer science.
It is well known that every such function
$f: \Fset_2^n \rightarrow \R$ can be represented as a linear
combination
\[
f = \sum_{S \subseteq [n]}{\hat{f}(S) \cdot \chi_S}
\]
of the $2^n$ functions $\{\chi_S\}_{S \subseteq [n]}$ defined
by $\chi_S(x) = (-1)^{\sum_{i \in S}{x_i}}$. This representation
is known as the {\em Fourier expansion} of the function $f$,
and the numbers $\hat{f}(S)$ are known as its \emph{Fourier coefficients}.
The Fourier expansion of functions plays a central role in analysis
of Boolean functions and finds applications in numerous areas of
theoretical computer science including learning theory, property testing,
hardness of approximation, social choice theory, and
cryptography. For an in-depth introduction to the topic the reader
is referred to the book of O'Donnell~\cite{O'Book14}.

A classical result in learning theory is a general algorithm due to Kushilevitz
and Mansour~\cite{KushilevitzM93}, based on results of Linial, Mansour, and Nisan~\cite{LinialMN93} and
Goldreich and Levin~\cite{GoldreichL89}, which enables to efficiently learn classes of Boolean
functions with a ``simple'' Fourier expansion. A common notion of simplicity of Fourier
expansion is its sparsity. A function is said to be {\em $k$-Fourier-sparse}
if it has at most $k$ nonzero Fourier coefficients. It follows
from~\cite{KushilevitzM93} that given {\em query}
access to a $k$-Fourier-sparse Boolean function
$f: \Fset_2^n \rightarrow \{0,1\}$
it is possible to estimate its Fourier coefficients and to get a
good approximation of $f$ in
running time polynomial in $n$ and $k$. Later, it was shown that
such running time even allows
to reconstruct the function $f$ exactly~\cite{GopalanOSSW11}.

In recent years, properties of the Fourier expansion of functions were
studied in the property testing framework. We now mention some of those
results; since this will not be needed for the sequel, the reader can
skip directly to the description of our results in the next section.
Gopalan, O'Donnell, Servedio, Shpilka, and Wimmer considered in~\cite{GopalanOSSW11} the problem of testing
if a given Boolean function is $k$-Fourier-sparse or $\eps$-far from any such function.
Another problem studied there is that of deciding if a function is
\emph{$k$-Fourier-dimensional}, that is, the Fourier support, viewed as a subset of $\Fset_2^n$,
spans a subspace of dimension at most $k$, or $\eps$-far from satisfying this property.
Gopalan et al.~\cite{GopalanOSSW11} established testers for these properties whose query
complexities depend only on $k$ and $\eps$. For $k$-Fourier-sparsity the query complexity
was a certain polynomial in $k$
and $1/\eps$ and for $k$-Fourier-dimensionality it was $O(k \cdot 2^{2k}/\eps)$.
They also proved lower bounds of $\Omega(\sqrt{k})$ and $\Omega(2^{k/2})$ respectively.
Another parameter associated with Boolean functions is the degree of its
representation as a polynomial over $\Fset_2$.
The algorithmic task of testing if a function has $\Fset_2$-degree
at most $d$ or is $\eps$-far
from any such function was considered by Alon et al.~\cite{AlonKKLR05} and then by
Bhattacharyya et al.~\cite{BhattacharyyaKSSZ10}, who proved
tight upper and lower bounds of $\Theta(2^d+1/\eps)$ on the
query complexity. Note that all the above properties fall
into the class of linear-invariant properties, i.e., properties that
are closed under compositions with any invertible linear transformation of the domain.
These properties have recently attracted a significant amount of attention
in the attempt to characterize efficient testability of them
(see~\cite{Sudan10,Bhattacharyya13} for related surveys).

\subsection{Our Results}

\paragraph{List-decoding size.}
Our main technical result from which we derive all other results
is concerned with the list-decoding size of Fourier-sparse Boolean functions.
In general, the list-decoding problem of an error correcting code for a distance parameter $d$
asks to find all the codewords whose Hamming distance from a given word is at most $d$. Here we
consider the (non-linear) binary code of block length $2^n$ whose codewords represent all the $k$-Fourier-sparse
Boolean functions on $n$ variables.

It is not difficult to show that the total number of such functions is at most $2^{O(nk)}$.
Indeed, there are $2^{O(nk)}$ ways to choose the support of $\hat{f}$, and $2^{O(nk)}$
ways to set those Fourier coefficients which must all be integer multiples of $2^{-n}$ in $[-1,+1]$.
It is also not difficult to show that the distance between any two distinct codewords
is at least $2^n/k$. Indeed, it is known that every $k$-Fourier-sparse
Boolean function has $\Fset_2$-degree $d \leq \log_2 k$ (see, e.g.,~\cite[Lemma~3]{BernasconiC99}), and therefore, by the Schwartz-Zippel lemma, every two distinct $k$-Fourier-sparse Boolean functions disagree on at least $1/k$ fraction of the inputs.
As a result, for every function $f: \Fset_2^n \rightarrow \R$ there
is at most one codeword of distance smaller than $2^{n}/(2k)$ from $f$.

We are not aware of any other known bounds beyond those two naive ones.
We address this question in the following theorem.

\begin{theorem}\label{thm:IntroList}
For every function $f : \Fset_2^n \rightarrow \R$, the number of $k$-Fourier-sparse Boolean functions of distance at most $d$ from $f$ is $2^{O(nd k \log k/2^n)}$.
\end{theorem}

We observe that for certain choices of $k$ and $d$ the bound given in Theorem~\ref{thm:IntroList} is tight.
For example, let $f$ be the constant zero function,
let $k<2^{0.9n}$ be a power of $2$, and take $d = 2^{n}/k$.
Consider all the indicator functions of linear subspaces of $\Fset_2^n$
of co-dimension $\log_2 k$. Every such function is of distance $d$ from
$f$ and is $k$-Fourier-sparse (see Claim~\ref{claim:affineSparse}). The
number of such functions is $2^{\Theta(n \log k)} = 2^{\Theta(n dk \log k /2^n)}$.

\paragraph{Learning from samples.}
As an application of the list-decoding bound, we next consider
the problem of learning the class of $k$-Fourier-sparse
Boolean functions on $n$ variables (exactly) from uniform and independent
random \emph{samples}~(see, e.g.,~\cite{AndoniPV014,KSDK14} for related work).
Let us note already at the outset that all the results mentioned here are \emph{not} efficient:
it is not known if there is an algorithm for the problem whose
running time is some fixed polynomial in $n$ times an arbitrary function of $k$.
Among other things, such an algorithm would imply a breakthrough on the long-standing open question
of learning juntas from samples~\cite{Blum03,MosselOS04,Valiant12,KSDK14}.

The question of recovering a function that
is sparse in the Fourier (or other) basis from a few samples
is the central question in the area of sparse recovery.
It has been intensely investigated for over a decade
and, among other things, has applications for compressed sensing and for the data stream model.
The best previously known bounds on our question
are $O(n \cdot k \log^3 k) \leq O(n^4 \cdot k)$ due to Cheraghchi, Guruswami, and Velingker~\cite{CheraghchiGV13}
and $O(n^2 \cdot k \log k) \leq O(n^3 \cdot k)$ due to Bourgain~\cite{Bourgain14},
improving on a previous bound of Rudelson and Vershynin~\cite{RudelsonV08}
(who themselves improved on the work of Cand{\`e}s and Tao~\cite{CandesT}).
We note in passing that they actually answer a harder question: first, because
they handle all functions, not necessarily Boolean-valued; second,
because they show that a randomly chosen set of sample locations
of the above cardinality is good with high probability simultaneously \emph{for all} $k$-Fourier-sparse
functions (sometimes known as the ``deterministic'' setting), whereas we only want a random set of
sample locations to be good with high probability for any \emph{fixed}
$k$-Fourier-sparse function (the ``randomized'' setting);
finally, because they obtain the recovery result by proving a
``restricted isometry property'' of the Fourier matrix
which among other things implies a recovery algorithm running in time polynomial in $2^n$ and $k$.

Using Theorem~\ref{thm:IntroList}, we improve the upper bound on the
sample complexity of learning Fourier-sparse Boolean functions.

\begin{corollary}\label{cor:IntroSample}
The number of uniform and independent random samples required for learning the class of $k$-Fourier-sparse Boolean functions on $n$ variables is $O(n \cdot k \log k)$.
\end{corollary}

\noindent
We believe that our better bound and its elementary proof shed more
light on the problem and might be useful elsewhere.
In fact, in a follow-up work~\cite{HavivR15rip} we employ the techniques developed here to study
the ``restricted isometry property'' of random submatrices of Fourier (and other) matrices,
improving on the aforementioned works~\cite{CheraghchiGV13,Bourgain14}.
We finally note that a lower bound of $\Omega(k \cdot (n- \log_2 k))$ on the sample complexity can be
obtained by considering the problem of learning indicator functions of affine subspaces of $\Fset_2^n$ of
co-dimension $\log_2 k$ (see Theorem~\ref{thm:LearningLower}; see, e.g.,~\cite{BaIPW10} for the same lower bound in a
different setting).

\paragraph{Testing Booleanity.}
We next consider the problem of testing Booleanity of Fourier-sparse functions, which was
introduced and studied by Gur and Tamuz in~\cite{GurT13}. In this problem,
given access to a $k$-Fourier-sparse function $f: \Fset_2^n \rightarrow \R$,
one has to decide if $f$ is Boolean,
i.e., its image is contained in $\{0,1\}$, or not. The objective is to distinguish between the
two cases with some constant probability using as few queries to $f$ as possible. It was shown
in~\cite{GurT13} that there exists a
(non-adaptive one-sided error) tester for the problem with query complexity $O(k^2)$, and that every tester for the problem has query complexity $\Omega(k)$. Here, we use our result on learning $k$-Fourier-sparse Boolean functions to improve the upper bound of~\cite{GurT13} and prove the following.

\begin{theorem}\label{thm:BoolUpperIntro}
For every $k$ there exists a non-adaptive one-sided error tester that using $O(k \cdot \log^2 k)$ queries to an input $k$-Fourier-sparse function $f: \Fset_2^n \rightarrow \R$ decides if $f$ is Boolean or not with constant success probability.
\end{theorem}

We note that, while the tester established in Theorem~\ref{thm:BoolUpperIntro}
has an improved query complexity, it is not clear if it is efficient with
respect to running time. It can be shown, though, that using
the learning algorithm of Fourier-sparse functions
that follows from~\cite{Bourgain14,HavivR15rip} (instead of Corollary~\ref{cor:IntroSample})
in our proof of Theorem~\ref{thm:BoolUpperIntro},
one can obtain an efficient algorithm (running in time polynomial in $n$ and $k$) with the slightly worse query complexity of $O(k \cdot \log^3 k)$.

Finally, we complement Theorem~\ref{thm:BoolUpperIntro} by the following nearly matching lower bound.

\begin{theorem}\label{thm:BoolLowerIntro}
Every non-adaptive one-sided error tester for Booleanity of $k$-Fourier-sparse functions has query complexity $\Omega(k \cdot \log k)$.
\end{theorem}

\subsection{Overview of Proofs}

\subsubsection{The List-Decoding Size of Fourier-Sparse Boolean Functions}

In order to prove Theorem~\ref{thm:IntroList}, we have to bound from above the number of $k$-Fourier-sparse Boolean functions of distance at most $d$ from a general function $f: \Fset_2^n \rightarrow \R$. In the discussion below, let us consider the special case where $f$ is the constant zero function. The general result follows easily.

Here, we have to bound the number of $k$-Fourier-sparse Boolean
functions $g: \Fset_2^n \rightarrow \{0,1\}$ of support size at most $d$.
We start by observing using Parseval's theorem that such functions
have small spectral norm $\|\hat{g}\|_1 = \sum_{S \subseteq [n]}{|\hat{g}(S)|}$.
Next, we observe that the Fourier expansion of the normalized function $g / \|\hat{g}\|_1$ is
a convex combination of functions $\pm \chi_S$, and thus can be viewed,
following a technique of Bruck and Smolensky~\cite{BruckS92}, as an expectation over a distribution on the $S$'s.
Using the Chernoff-Hoeffding bound and the bound on the spectral norm, we obtain a succinct
representation for every such function $g$. The ability to represent these functions by
a binary string of bounded length yields the upper bound on their number. We note that the proof approach
somewhat resembles that of the upper bound on the list-decoding
size of Reed-Muller codes due to Kaufman, Lovett, and Porat~\cite{KaufmanLP12}.

\subsubsection{Learning Fourier-Sparse Boolean Functions}

As a warmup, let us mention
an easy upper bound of $O(n \cdot k^2)$. This follows by
recalling that there are at most $2^{O(nk)}$ $k$-Fourier-sparse Boolean functions,
and that each one differs from any fixed function on at least $1/k$ fraction of the inputs.
Hence by the union bound, after $O(n \cdot k^2)$ samples all other functions will
be eliminated.

The improved bound in Corollary~\ref{cor:IntroSample} follows similarly
using the list-decoding result of Theorem~\ref{thm:IntroList}.
Namely, we apply the union bound separately on functions of different
distances from the input function. Functions that are nearby are harder to
``hit'' using random samples, but by the theorem, there are few of them; functions
that are further away are in abundance, but they are easier to ``hit'' using random samples.

\subsubsection{Testing Booleanity of Fourier-Sparse Functions}\label{subsec:Bool}

The testing Booleanity problem is somewhat different from typical property
testing problems. Indeed, in property testing one usually has to distinguish
objects that satisfy a certain property from those that are $\eps$-far from
the property for some distance parameter $\eps >0$. However, here the
tester is required to decide if the function satisfies the Booleanity
property or not, with no distance parameter involved. This unusual
setting makes sense in this case because Fourier-sparse non-Boolean
functions are always quite far from every Boolean function. More
precisely, the authors of~\cite{GurT13} used the uncertainty
principle (see Proposition~\ref{prop:uncertain}) to prove that
every $k$-Fourier-sparse non-Boolean function
$f: \Fset_2^n \rightarrow \R$ is non-Boolean on
at least $\Omega(2^n/k^2)$ inputs
(see Claim~\ref{claim:non-zero}).
This immediately implies a (non-adaptive one-sided error) tester that uses $O(k^2)$
queries: just check that $f$ is Boolean on $O(k^2)$ uniform inputs in $\Fset_2^n$.

The analysis of~\cite{GurT13} turns out to be tight, as there are $k$-Fourier-sparse non-Boolean functions that are not Boolean at
only $\Theta(2^n/k^2)$ points. Indeed, for an even integer $n$, consider
the function $f:\Fset_2^n \to \{0,1,2\}$ defined by
\begin{align}\label{no_func}
f(x_1,\ldots,x_n) = \myand(x_1,\ldots,x_{n/2})+\myand(x_{n/2+1},\ldots,x_n),
\end{align}
which is not Boolean at only one point and has Fourier-sparsity $2 \cdot 2^{n/2}$ (see Claim~\ref{claim:affineSparse}).

\paragraph{Upper bound.}
We prove Theorem~\ref{thm:BoolUpperIntro} using our learning result, Corollary~\ref{cor:IntroSample}.
To do so, we first observe that a restriction of a $k$-Fourier-sparse non-Boolean function to a random subspace of dimension $O(\log k)$ is non-Boolean with high probability (see Lemma~\ref{lemma:restriction}). Since a restriction to a subspace does not increase the Fourier-sparsity, this reduces our problem to testing Booleanity of $k$-Fourier-sparse functions on $n = O(\log k)$ variables. Then, after $O(k \cdot \log^2 k)$ samples from the subspace, if a non-Boolean value was found then we are clearly done. Otherwise, by Corollary~\ref{cor:IntroSample}, the samples uniquely specify a Boolean candidate for the restricted function. Such a function must be quite far from every other $k$-Fourier-sparse function (Boolean or not; see Claim~\ref{claim:sparse_dist}). This enables us to decide if the restricted function equals the Boolean candidate function or not.

\paragraph{Lower bound.}
The upper bound in Theorem~\ref{thm:BoolUpperIntro} gets close to the $\Omega(k)$ lower bound proven by Gur and Tamuz in~\cite{GurT13}. For their lower bound, they considered the following two distributions: (a) the uniform distribution over all Boolean $n$-variable functions that depend only on their first $\log_2 k$ variables; (b) the uniform distribution over all $n$-variable functions that depend only on their first $\log_2 k$ variables and return a Boolean value on $k-1$ of the assignments to the relevant variables and the value $2$ otherwise. It can be easily seen that any (possibly adaptive) tester that distinguishes with some constant probability between distributions (a) and (b) has query complexity $\Omega(k)$. Since the first distribution is supported on $k$-Fourier-sparse Boolean functions and the second on $k$-Fourier-sparse non-Boolean functions, this implies that the same lower bound holds for the query complexity of testing Booleanity of $k$-Fourier-sparse functions.

Note that the distributions considered above are supported on
$\log_2 k$-Fourier-dimensional functions.
It can be seen (say, using the uncertainty principle) that such functions are
not Boolean on at least $1/k$ fraction of their inputs, so $O(k)$
random samples suffice for finding a non-Boolean value if exists.
Hence, in order to get beyond the
$\Omega(k)$ lower bound, we need to consider $k$-Fourier-sparse
functions that are not Boolean at only $o(1/k)$ fraction of the inputs --
our functions will actually have $O(1/k^2)$ fraction of such inputs.

Specifically, we consider the distribution of functions obtained by composing the function $f$ given in~(\ref{no_func}) with a random invertible affine transformation. This is the class of functions that can be represented as a sum $\mathbb{1}_{V_1} + \mathbb{1}_{V_2}$ of two indicators of affine subspaces $V_1,V_2 \subseteq \Fset_2^n$ of dimension $n/2$, which intersect at exactly one point. Intuitively, it seems that distinguishing the functions in this class from those where $V_1$ and $V_2$ have empty intersection requires the tester to learn the affine subspaces $V_1$ and $V_2$, a task that requires $\Omega(n \cdot 2^{n/2})$ queries. We prove such a lower bound for non-adaptive one-sided error testers. Since the above functions are $k$-Fourier-sparse for $k=O(2^{n/2})$, the obtained lower bound is $\Omega(k \cdot \log k)$.

\section{Preliminaries}

Let $[n]$ denote the set $\{1,\ldots,n\}$. A function $f: \Fset_2^n \rightarrow \R$ is {\em Boolean} if its image is contained in $\{0,1\}$ and is {\em non-Boolean} otherwise. The {\em distance} between two functions $f,g:\Fset_2^n \rightarrow \R$, denoted $\dist(f,g)$, is the number of vectors $x \in \Fset_2^n$ for which $f(x) \neq g(x)$.

\subsection*{Fourier Expansion}
For every $S \subseteq [n]$, let $\chi_S : \Fset_2^n \rightarrow \{-1,1\}$ denote the function defined by $\chi_S(x) = (-1)^{\sum_{i \in S}{x_i}}$. It is well known that the $2^n$ functions $\{\chi_S\}_{S \subseteq [n]}$ form an orthonormal basis of the space of functions $\Fset_2^n \rightarrow \R$ with respect to the inner product $\langle f,g \rangle = \expectation_{x}[f(x) \cdot g(x)]$, where $x$ is distributed uniformly over $\Fset_2^n$. Thus, every function $f: \Fset_2^n \rightarrow \R$ can be uniquely represented as a linear combination $f = \sum_{S \subseteq [n]}{\hat{f}(S) \cdot \chi_S}$ of this basis. This representation is called the {\em Fourier expansion} of $f$, and the numbers $\hat{f}(S)$ are referred to as its {\em Fourier coefficients}. The support of $f$ is defined by $\supp(f) = \{ x \in \Fset_2^n \mid f(x) \neq 0\}$ and the support of $\hat{f}$, known as the {\em Fourier spectrum} of $f$, by $\supp(\hat{f}) = \{ S \subseteq [n] \mid \hat{f}(S) \neq 0\}$. We say that $f$ is {\em $k$-Fourier-sparse}\footnote{Boolean functions are sometimes defined in the literature with range $\{-1,+1\}$ rather than $\{0,1\}$. Notice that this affects the Fourier-sparsity by at most $1$.} if $|\supp(\hat{f})| \leq k$. For every $p \geq 1$ we denote $\|\hat{f}\|_p = (\sum_{S \subseteq [n]}{|\hat{f}(S)|^p})^{1/p}$. For $p=1$, $\|\hat{f}\|_1$ is known as the {\em spectral norm} of $f$. {\em Parseval's theorem} states that $\expectation_{x}[f(x)^2] = \|\hat{f}\|_2^2$.

The uncertainty principle says that there is no nonzero function $f$ for which the supports of both $f$ and $\hat{f}$ are small (see, e.g.,~\cite[Exercise~3.15]{O'Book14}). We state it below with two simple consequences.

\begin{proposition}[The Uncertainty Principle]\label{prop:uncertain}
For every nonzero function $f: \Fset_2^n \rightarrow \R$,
$$|\supp(f)| \cdot |\supp(\hat{f})| \geq 2^n.$$
\end{proposition}

\begin{claim}\label{claim:sparse_dist}
For every two distinct $k$-Fourier-sparse functions $f,g:\Fset_2^n \rightarrow \R$, $\dist(f,g) \geq 2^n/(2k)$.
\end{claim}
\begin{proof}
Apply Proposition~\ref{prop:uncertain} to the function $f-g$, whose Fourier-sparsity is at most $2k$.
\end{proof}

\begin{claim}[\cite{GurT13}]\label{claim:non-zero}
For every $k$-Fourier-sparse function $f:\Fset_2^n \rightarrow \R$, if $f$ is non-Boolean then
$$|\{ x \in \Fset_2^n \mid f(x) \notin \{0,1\}\}| \geq \frac{2}{k^2+k+2} \cdot 2^n.$$
\end{claim}
\begin{proof}
Apply Proposition~\ref{prop:uncertain} to the function $f \cdot(f-1)$, whose Fourier-sparsity is at most
\[  |\{ S \triangle T \mid S,T \in \supp(\hat{f})\} | + |\supp(\hat{f})| \leq \binom{k}{2} + k +1,\]
where $\triangle$ stands for symmetric difference of sets.
\end{proof}

We also need the following simple claim.
\begin{claim}\label{claim:affineSparse}
For every affine subspace $V \subseteq \Fset_2^n$ of co-dimension $k$, the indicator function $\mathbb{1}_V: \Fset_2^n \rightarrow \{0,1\}$ is $2^k$-Fourier-sparse.
\end{claim}

\begin{proof}
Since $V$ has co-dimension $k$, there exist $a_1,\ldots,a_k \in \Fset_2^n$ and $b_1,\ldots,b_k \in \Fset_2$ such that $V = \{ x \in \Fset_2^n \mid \langle x, a_i \rangle = b_i,~i=1,\ldots,k \}$. For every $i$, let $S_i \subseteq [n]$ denote the set whose characteristic vector is $a_i$, and observe that for every $x \in \Fset_2^n$,
$$\mathbb{1}_V (x) = \prod_{i=1}^{k}{\Big (\frac{1+(-1)^{b_i} \cdot \chi_{S_i}(x)}{2} \Big)}.$$
This representation implies that $\mathbb{1}_V$ is $2^k$-Fourier-sparse.
\end{proof}

\subsection*{Chernoff-Hoeffding Bound}
\begin{theorem}\label{thm:Chernoff}
Let $X_1,\ldots,X_N$ be $N$ identically distributed independent random variables in $[-a,+a]$ satisfying $\expectation[X_i]=\mu$ for all $i$. Then for every $\delta \leq 1/2$ and $N \geq C \cdot a^2 \cdot \log(1/\delta)/\eps^2$, for a universal constant $C$, it holds that
$$\Prob{}{\Big|\mu - \frac{1}{N} \cdot \sum_{i=1}^{N}{X_i} \Big| < \eps} \geq 1- \delta.$$
\end{theorem}

\section{The List-Decoding Size of Fourier-Sparse Boolean Functions}

We turn to prove Theorem~\ref{thm:IntroList}, which provides an upper bound on the list-decoding size of the code of block length $2^n$ of all $k$-Fourier-sparse Boolean functions on $n$ variables. Equivalently, for a general distance $d$ and a function $f : \Fset_2^n \rightarrow \R$ we bound the number of $k$-Fourier-sparse Boolean functions on $n$ variables of distance at most $d$ from $f$.

We start by proving that a function $f : \Fset_2^n \rightarrow \R$ with small spectral norm can be well approximated by a linear combination of few functions from $\{ \chi_S \}_{S \subseteq [n]}$ with coefficients of equal magnitude. This was essentially proved in~\cite{BruckS92} and we include here the proof for completeness.
\begin{lemma}\label{lemma:f_rep}
For every function $f : \Fset_2^n \rightarrow \R$, $\eps >0$, and $\delta \in (0,1/2]$,
there exists a collection\footnote{Repetitions of subsets in the collection $\calF$ are allowed.}
$\calF$ of $O(\|\hat{f}\|_1^2 \cdot \log(1/\delta) /\eps^2)$ subsets of $[n]$
with signs $(a_S \in \{\pm 1\})_{S \in \calF}$ such that for all but at most
$\delta$ fraction of $x \in \Fset_2^n$ it holds that
\[
\Big|f(x) - \frac{\|\hat{f}\|_1}{|\calF|} \cdot \sum_{S \in \calF}{a_S \cdot \chi_S(x)}\Big| < \eps \;.
\]
\end{lemma}

\begin{proof}
Observe that the function $f$ can be represented as follows.
\[
f = \sum_{S \subseteq [n]}{\hat{f}(S)\cdot \chi_S} = \sum_{S \subseteq [n]}{\frac{|\hat{f}(S)|}{\|\hat{f}\|_1} \cdot \|\hat{f}\|_1 \cdot \sign(\hat{f}(S)) \cdot \chi_S} = \expectation_{S \sim D}[\|\hat{f}\|_1 \cdot \sign(\hat{f}(S)) \cdot \chi_S],
\]
where $D$ is the distribution defined by $D(S) = |\hat{f}(S)|/\|\hat{f}\|_1$.
Let $\calF$ be a collection of
$|\calF| = O(\|\hat{f}\|_1^2 \cdot \log(1/\delta)/\eps^2)$
independent random samples from the distribution $D$. For every $x \in \Fset_2^n$,
the Chernoff-Hoeffding bound (Theorem~\ref{thm:Chernoff}) implies that
with probability at least $1-\delta$ it holds that
\begin{eqnarray}\label{Eq:g_approx}
\Big|f(x) - \frac{1}{|\calF|} \cdot \sum_{S \in \calF}{\|\hat{f}\|_1 \cdot a_S \cdot \chi_S(x)} \Big| < \eps,
\end{eqnarray}
where $a_S = \sign(\hat{f}(S))$. By linearity of expectation, it follows that there exist $\calF$ and signs $(a_S)_{S \in \calF}$ for which~\eqref{Eq:g_approx} holds for all but at most $\delta$ fraction of $x\in \Fset_2^n$, as required.
\end{proof}

We now apply Lemma~\ref{lemma:f_rep} to Fourier-sparse functions in $\Fset_2^n \rightarrow \{-1,0,+1\}$ with bounded support size, and then, in Corollary~\ref{cor:weight}, derive an upper bound on the number of these functions.

\begin{corollary}\label{cor:f_rep}
Let $f : \Fset_2^n \rightarrow \{-1,0,+1\}$ be a $k$-Fourier-sparse function satisfying $|\supp(f)| \leq d$. Then for every $\delta \in (0,1/2]$ there exists a collection $\calF$ of $O(dk \log(1/\delta)/2^n)$ subsets of $[n]$ with signs $(a_S \in \{\pm 1\})_{S \in \calF}$ such that for all but at most $\delta$ fraction of $x \in \Fset_2^n$ it holds that
\[
\Big|f(x) - \frac{\|\hat{f}\|_1}{|\calF|} \cdot \sum_{S \in \calF}{a_S \cdot \chi_S(x)}\Big| < \frac{1}{2} \; .
\]
\end{corollary}

\begin{proof}
By the Cauchy-Schwarz inequality and Parseval's theorem, we obtain that
$$\frac{\|\hat{f}\|_1^2}{k} \leq \sum_{S \subseteq [n]}{\hat{f}(S)^2} = 2^{-n} \cdot \sum_{x \in \Fset_2^n}{f(x)^2} \leq \frac{d}{2^n}.$$
The corollary follows from Lemma~\ref{lemma:f_rep}, applied with $\eps=1/2$, for $|\calF| = O(\|\hat{f}\|_1^2 \log(1/\delta)/\eps^2) = O(dk \log(1/\delta) / 2^n)$.
\end{proof}

\begin{corollary}\label{cor:weight}
The number of $k$-Fourier-sparse functions $f : \Fset_2^n \rightarrow \{-1,0,+1\}$ satisfying $|\supp(f)| \leq d$ is $2^{O(nd k \log k /2^n)}$.
\end{corollary}

\begin{proof}
For every $k$-Fourier-sparse function $f : \Fset_2^n \rightarrow \{-1,0,+1\}$
satisfying $|\supp(f)| \leq d$, let $\calF$ and $(a_S)_{S \in \calF}$ be as
given by Corollary~\ref{cor:f_rep} for, say, $\delta = 1/(5k)$. Since the
range of $f$ is $\{-1,0,+1\}$, it follows that the collection $\calF$, the
signs $(a_S)_{S \in \calF}$, and the value of $\|\hat{f}\|_1$ define a
function of distance at most $\delta \cdot 2^n$ from $f$. Notice that by
Claim~\ref{claim:sparse_dist} and our choice of $\delta$, the distance between
every two distinct $k$-Fourier-sparse functions is larger than $2\delta \cdot 2^n$.
Thus, a function of distance at most $\delta \cdot 2^n$ from $f$ fully
defines $f$. This implies that $f$ can be represented by a
binary string of length $O(n \cdot d k \log k/2^n)$, so the total number of
such functions is $2^{O(nd k \log k /2^n)}$.
\end{proof}

The bound in Corollary~\ref{cor:weight} implies a bound on the number of Fourier-sparse Boolean functions of bounded distance from a given Boolean function.

\begin{corollary}\label{cor:list}
For every $k$-Fourier-sparse Boolean function $f : \Fset_2^n \rightarrow \{0,1\}$, the number of $k$-Fourier-sparse Boolean functions of distance at most $d$ from $f$ is $2^{O(nd k \log k /2^n)}$.
\end{corollary}

\begin{proof}
Let $f : \Fset_2^n \rightarrow \{0,1\}$ be a $k$-Fourier-sparse Boolean function. Consider the mapping that maps every $k$-Fourier-sparse Boolean function $g : \Fset_2^n \rightarrow \{0,1\}$, whose distance from $f$ is at most $d$, to the function $h = f-g$. Observe that $h$ is a $2k$-Fourier-sparse function from $\Fset_2^n$ to $\{-1,0,+1\}$ satisfying $|\supp(h)| \leq d$. By Corollary~\ref{cor:weight}, the number of such functions $h$ is bounded by $2^{O(nd k \log k /2^n)}$. Since the above mapping is bijective, this bound holds for the number of functions $g$ as well.
\end{proof}

Equipped with Corollary~\ref{cor:weight}, we restate and prove Theorem~\ref{thm:IntroList}.

\newtheorem*{thma}{Theorem~\ref{thm:IntroList}}
\begin{thma}
For every function $f : \Fset_2^n \rightarrow \R$, the number of $k$-Fourier-sparse Boolean functions of distance at most $d$ from $f$ is $2^{O(nd k \log k /2^n)}$.
\end{thma}

\begin{proof}
If there is no $k$-Fourier-sparse Boolean function of distance at most $d$ from $f$, then the bound trivially holds. So assume that such a function $g : \Fset_2^n \rightarrow \{0,1\}$ exists. Observe that every $k$-Fourier-sparse Boolean function of distance at most $d$ from $f$ has distance at most $2d$ from $g$. Thus, by Corollary~\ref{cor:list} applied to $g$, the number of such functions is at most $2^{O(n d k \log k /2^n)}$.
\end{proof}

\subsection{The Sample Complexity of Learning Fourier-Sparse Boolean Functions}

The {\em sample complexity} of learning a class of functions is the minimum number of uniform and independent random samples needed from a function in the class for specifying it with high success probability. Here we consider the class of $k$-Fourier-sparse Boolean functions on $n$ variables, and show how Theorem~\ref{thm:IntroList} implies an upper bound on the sample complexity of learning it (Corollary~\ref{cor:learn_Bool}).

\begin{theorem}\label{thm:learning}
For every $n$, $1 < k \leq 2^n$, and a $k$-Fourier-sparse function $f: \Fset_2^n \rightarrow \R$, the following holds. The probability that when sampling $O(n \cdot k \log k)$ uniform and independent random samples from $f$, there exists a $k$-Fourier-sparse Boolean function $g \neq f$ that agrees with $f$ on all the samples is $2^{-\Omega(n \log k)}$.
\end{theorem}

\begin{proof}
Consider $q = O(nk \log k)$ samples $(x,f(x))$ from a $k$-Fourier-sparse function $f: \Fset_2^n \rightarrow \R$, where $x$ is distributed uniformly and independently in $\Fset_2^n$. By Claim~\ref{claim:sparse_dist}, the distance between $f$ and every other $k$-Fourier-sparse function is at least $2^{n}/(2k)$. For an integer $\ell \in [1, \lfloor \log_2 2k \rfloor]$, consider all the $k$-Fourier-sparse Boolean functions whose distance from $f$ is in $[2^{n-\ell},2^{n-\ell+1}]$. By Theorem~\ref{thm:IntroList}, the number of such functions is $2^{O(n k \log k/2^{\ell})}$. The probability that such a function agrees with $q$ random independent samples of $f$ is at most $(1-2^{-\ell})^q$. By the union bound, the probability that at least one of these functions agrees with the $q$ samples is at most
$$2^{O(n k \log k /2^{\ell})} \cdot (1-2^{-\ell})^q \leq 2^{O(n k \log k /2^{\ell})} \cdot e^{-q/2^{\ell}} \leq 2^{-\Omega(n \log k)},$$
where the last inequality holds for an appropriate choice of $q = O(n k \log k)$. By applying the union bound over all the values of $\ell$, it follows that with probability $1-2^{-\Omega(n \log k)}$ all the $k$-Fourier-sparse Boolean functions (besides $f$) are eliminated, completing the proof.
\end{proof}

The following corollary follows immediately from Theorem~\ref{thm:learning} and confirms Corollary~\ref{cor:IntroSample}.

\begin{corollary}\label{cor:learn_Bool}
For every $n$ and $1 \leq k \leq 2^n$, the number of uniform and independent random samples required for learning the class of $k$-Fourier-sparse Boolean functions on $n$ variables with success probability $1-2^{-\Omega(n \log k)}$ is $O(n \cdot k \log  k)$.
\end{corollary}

We end with the following simple lower bound.
\begin{theorem}\label{thm:LearningLower}
For every $n$ and $1 \leq k \leq 2^n$, the number of uniform and independent random samples required for learning the class of $k$-Fourier-sparse Boolean functions on $n$ variables with constant success probability is $\Omega(k \cdot (n-\log_2 k))$.
\end{theorem}

\begin{proof}
Assume without loss of generality that $k$ is a power of $2$.
Let $A$ be an algorithm for learning the class above with constant success probability $p>0$ using $q$ uniform and independent random samples. Consider the class $\calG$ of indicators of affine subspaces of $\Fset_2^n$ of co-dimension $\log_2 k$ (i.e., affine subspaces of $\Fset_2^n$ of size $2^n/k$). By  Claim~\ref{claim:affineSparse}, the functions in $\calG$ are $k$-Fourier-sparse. Observe that their number satisfies
\[ |\calG| = 2^{\Theta(n \cdot \min(\log_2 k, n-\log_2 k))}.\]
By Yao's minimax principle, there exists a deterministic algorithm $A'$ (obtained by fixing the random coins of $A$) that given evaluations of a function, chosen uniformly at random from $\calG$, on a {\em fixed} collection of $q$ points in $\Fset_2^n$, learns it with success probability $p$.

Now, observe that the expected number of $1$-evaluations that $A'$ receives is $q/k$. By Markov's inequality, the probability that $A'$ receives at least $2q/(pk)$ $1$-evaluations is at most $p/2$. It follows that for at least $p/2$ fraction of the functions in $\calG$ the algorithm $A'$ receives at most $2q/(pk)$ $1$-evaluations and learns them correctly. Assuming that $pk \geq 2$, the number of possible evaluation sequences on these inputs is at most
\[\sum_{i=0}^{2q/(pk)}\binom{q}{i} \leq (k \cdot pe/2)^{2q/(pk)} \leq 2^{O(q \cdot \log_2 k /k)},\]
where for the first inequality we used the standard inequality $\sum_{i=0}^{t}\binom{q}{i} \leq (qe/t)^t$ which holds for $t \leq q$ (see, e.g.,~\cite[Proposition~1.4]{JuknaEC}).
The above is bounded from below by $|\calG| \cdot p/2$, implying that
\[q \geq \Omega( n \cdot \min(\log_2 k, n-\log_2 k) \cdot k /\log_2 k) \geq \Omega(k \cdot (n-\log_2 k)),\]
where the last inequality follows by considering separately the cases of $k \geq 2^{n/2}$ and $k < 2^{n/2}$.
In case that $pk < 2$, the number of possible evaluation sequences is at most $2^q$, and the bound follows similarly using the assumption that $p$ is a fixed constant.
\end{proof}

\section{Testing Booleanity of Fourier-Sparse Functions}

In this section we prove upper and lower bounds on the query complexity of testing Booleanity of Fourier-sparse functions. For a parameter $k$, consider the problem in which given access to a $k$-Fourier-sparse function $f : \Fset_2^n \rightarrow \R$ one has to decide if $f$ is Boolean, i.e., $f(x) \in \{0,1\}$ for every $x \in \Fset_2^n$, or not, with some constant success probability.

\subsection{Upper Bound}

As mentioned before, Gur and Tamuz proved in~\cite{GurT13} that every $k$-Fourier-sparse non-Boolean function $f$ on $n$ variables satisfies $f(x) \notin \{0,1\}$ for at least $\Omega(2^n/k^2)$ inputs $x \in \Fset_2^n$ (see Claim~\ref{claim:non-zero}). Thus, querying the input function $f$ on $O(k^2)$ independent and random inputs suffices in order to catch a non-Boolean value of $f$ if such a value exists. In the following lemma it is shown that it is not really needed to choose the $O(k^2)$ random vectors {\em independently}. It turns out that a restriction of a $k$-Fourier-sparse non-Boolean function to a random linear subspace of size $O(k^2)$, that is, of dimension $\approx 2\log_2 k$, is with high probability non-Boolean. Thus, the tester could randomly pick such a subspace and query $f$ on all of its vectors. This decreases the amount of randomness used in the tester of~\cite{GurT13} from $O(nk^2)$ to $O(n \log k)$. More importantly for us, this reduces the problem of testing Booleanity of $k$-Fourier-sparse functions on $n$ variables to the case of $k = \Theta(2^{n/2})$.

\begin{lemma}\label{lemma:restriction}
Let $f:\Fset_2^n \rightarrow \R$ be a $k$-Fourier-sparse non-Boolean function, and denote $L = (k^2+k+2)/2$. Then, for every $\delta > 0$, the restriction of $f$ to a uniformly chosen random linear subspace of dimension $r \geq \log_2 (L/\delta)$ is also non-Boolean with probability at least $1-\delta$.
\end{lemma}

\begin{proof}
Let $f:\Fset_2^n \rightarrow \R$ be a $k$-Fourier-sparse non-Boolean function. By Claim~\ref{claim:non-zero}, there are at least $2^n/L$ vectors $x \in \Fset_2^n$ for which $f(x) \notin \{0,1\}$. This implies that there exists a set $S$ of at least $\log_2(2^n/L)$ {\em linearly independent} vectors in $\Fset_2^n$ on which $f$ is not Boolean. Consider a linear subspace $V \subseteq \Fset_2^n$ of dimension $n-1$ chosen uniformly at random. Since the vectors in $S$ are linearly independent, the probability that no vector in $S$ is in $V$ is $2^{-|S|} \leq \frac{L}{2^n}$. It follows that the restriction $f|_V$ of $f$ to $V$ is a $k$-Fourier-sparse function defined on a linear subspace of dimension $n-1$, and its probability to be Boolean is at most $\frac{L}{2^n}$. Note that one can think of the domain of $f|_V$ as $\Fset_2^{n-1}$, because $V$ and $\Fset_2^{n-1}$ are isomorphic and a composition with an invertible linear transformation does not affect the Fourier-sparsity. Now, let us repeat the above process $n-r-1$ additional times, until we get a linear subspace of dimension $r$. The probability that the function becomes Boolean in one of the steps is at most
$$\frac{L}{2^n} + \frac{L}{2^{n-1}}+\cdots+\frac{L}{2^{r+1}} \leq \frac{L}{2^r} \leq \delta,$$
and we are done.
\end{proof}

We now restate and prove Theorem~\ref{thm:BoolUpperIntro}, which gives an upper bound of $O(k \cdot \log^2 k)$ on the query complexity of testing Booleanity of $k$-Fourier-sparse functions. In the proof, we first apply Lemma~\ref{lemma:restriction} to restrict the input function to a subspace of dimension $O(\log k)$. Then, we apply Theorem~\ref{thm:learning} in an attempt to learn the restricted function and check if it is consistent with some $k$-Fourier-sparse Boolean function.

\newtheorem*{thmb}{Theorem~\ref{thm:BoolUpperIntro}}
\begin{thmb}
For every $k$ there exists a non-adaptive one-sided error tester that using $O(k \cdot \log^2 k)$ queries to an input $k$-Fourier-sparse function $f: \Fset_2^n \rightarrow \R$ decides if $f$ is Boolean or not with constant success probability.
\end{thmb}

\begin{proof}
Consider the tester that given access to an input $k$-Fourier-sparse function $f: \Fset_2^n \rightarrow \R$ acts as follows:
\begin{enumerate}
  \item\label{itm:1} Pick uniformly at random a linear subspace $V$ of $\Fset_2^n$ of dimension $r = \min (n, \lceil\log_2(100L)\rceil)$, where $L = (k^2+k+2)/2$, and let $T$ be an invertible linear transformation mapping $\Fset_2^r$ to $V$.
  \item\label{itm:2} Query $f$ on $O(r \cdot k \log k)$ random vectors chosen uniformly and independently from the subspace $V$. Note that these queries can be seen as uniform and independent random samples from the function $g: \Fset_2^r \rightarrow \R$ defined as $g = f \circ T$.
  \item\label{itm:3} If there exists a $k$-Fourier-sparse Boolean function on $r$ variables that agrees with the above samples of $g$ then accept, and otherwise reject.
\end{enumerate}

We turn to prove the correctness of the above tester. If $f$ is a $k$-Fourier-sparse Boolean function then so is $g$, because a restriction to a subspace and a composition with a linear transformation leave the function $k$-Fourier-sparse and Boolean. Hence, in this case the tester accepts with probability $1$.

On the other hand, if $f$ is a $k$-Fourier-sparse non-Boolean function, then by
Lemma~\ref{lemma:restriction} the restriction of $f$ to the random subspace $V$
of dimension $r$ picked in Item~\ref{itm:1}, as well as the function $g$ defined
in Item~\ref{itm:2}, are also non-Boolean with probability at least $0.99$. In
this case, by Theorem~\ref{thm:learning}, the probability that there is a
$k$-Fourier-sparse Boolean function on $r$ variables that agrees with
$O(r \cdot k \log k)$ uniform and independent random samples from $g$ is
$2^{-\Omega(r \log k)}$, thus the tester correctly rejects with probability
at least, say, $0.9$, as required. Finally, observe that the number of
queries made by the tester is $O(r \cdot k \log k) = O(k \cdot \log^2 k)$.
\end{proof}

\subsection{Lower Bound}

We turn to restate and prove our lower bound on the query complexity of testing Booleanity of $k$-Fourier-sparse functions.

\newtheorem*{thmc}{Theorem~\ref{thm:BoolLowerIntro}}
\begin{thmc}
Every non-adaptive one-sided error tester for Booleanity of $k$-Fourier-sparse functions has query complexity $\Omega(k \cdot \log k)$.
\end{thmc}

\begin{proof}
For a given integer $k$, let $n$ be the largest even integer satisfying $k \geq 3 \cdot 2^{n/2}$. Define a distribution $D_{no}$ over functions in $\Fset_2^n \rightarrow \{0,1,2\}$ as follows. Pick uniformly at random a pair $(V_1,V_2)$ of affine subspaces satisfying $\dim(V_1)=\dim(V_2)=n/2$ and $|V_1 \cap V_2|=1$, and output the sum of indicators $\mathbb{1}_{V_1} + \mathbb{1}_{V_2}$. Notice that, by Claim~\ref{claim:affineSparse}, such a function has Fourier-sparsity at most $2 \cdot 2^{n/2} \leq k$. Thus, a function chosen from $D_{no}$ is $k$-Fourier-sparse and non-Boolean with probability $1$.

Let $T$ be a non-adaptive one-sided error randomized tester for Booleanity
of $k$-Fourier-sparse functions with query complexity $q$ and success
probability at least $2/3$. By Yao's minimax principle, there exists a
deterministic tester $T'$ (obtained by fixing the random coins of $T$)
that rejects a random
function chosen from $D_{no}$ with probability at least $2/3$.
Since $T$ is non-adaptive and has one-sided
error, it follows that $T'$ queries an input function on $q$ {\em fixed}
vectors $a_1,\ldots,a_q \in \Fset_2^n$, accepts every $k$-Fourier-sparse
Boolean function, and rejects a function chosen from $D_{no}$ with
probability at least $2/3$. We turn to prove that
$q > (n \cdot 2^{n/2}) / 1000 = \Omega(k \cdot \log k)$.

Assume in contradiction that $q \leq (n \cdot 2^{n/2}) / 1000$. Let $f$ be a random function chosen from $D_{no}$, that is, $f = \mathbb{1}_{V_1} + \mathbb{1}_{V_2}$ for random affine subspaces $V_1$ and $V_2$ of dimension $n/2$ satisfying $|V_1 \cap V_2| = 1$. For $i=1,2$, let $W_i$ be the affine span of $\{a_1,\ldots,a_q\} \cap V_i$. Let $E$ be the event that the intersection of $W_1$ and $W_2$ is empty. We turn to prove that if the event $E$ happens then the tester $T'$ accepts the function $f$ and that the probability of this event is at least $0.9$. This contradicts the success probability of $T'$ on functions chosen from $D_{no}$ and completes the proof.

\begin{lemma}
If the event $E$ happens then the tester $T'$ accepts the function $f$.
\end{lemma}

\begin{proof}
Assume that the event $E$ happens, i.e., $W_1 \cap W_2 = \emptyset$. Then, there exists an affine subspace $V'_2$ of dimension $n/2-1$ satisfying $W_2 \subseteq V'_2 \subsetneq V_2$ and $V_1 \cap V'_2 = \emptyset$. Consider the function $g = \mathbb{1}_{V_1} + \mathbb{1}_{V'_2}$. By Claim~\ref{claim:affineSparse}, $g$ is a Boolean function whose Fourier-sparsity is at most $3 \cdot 2^{n/2} \leq k$, thus it is accepted by $T'$. However, $g$ satisfies $g(a_i) = f(a_i)$ for every $1 \leq i \leq q$. This implies that $T'$ cannot distinguish between $g$ and $f$, so it must accept $f$ as well.
\end{proof}

\begin{lemma}
The probability of the event $E$ is at least $0.9$.
\end{lemma}

\begin{proof}
Denote by $X$ the number of vectors in $\{a_1,\ldots,a_q\} \cap V_1$. Since $V_1$ is distributed uniformly over all affine subspaces of dimension $n/2$, the probability that $a_i$ belongs to $V_1$ is $2^{-n/2}$ for every $1 \leq i \leq q$ . Thus, by linearity of expectation, $$\expectation[X] = \frac{q}{2^{n/2}} \leq \frac{(n \cdot 2^{n/2}) / 1000}{2^{n/2}} = \frac{n}{1000}.$$
By Markov's inequality, we obtain that
$$\Prob{}{\dim(W_1) \geq \frac{n}{10}} \leq \Prob{}{X \geq \frac{n}{10}} \leq \frac{1}{100}.$$

Now, fix a choice of $V_1$ for which $\dim(W_1) < n/10$, and consider the randomness over the choice of $V_2$. Notice that, conditioned on $V_1$, $V_2$ is distributed uniformly over all the affine subspaces of dimension $n/2$ which contain exactly one vector from $V_1$. By symmetry, every vector of $V_1$ has probability $|V_1|^{-1} = 2^{-n/2}$ to belong to $V_2$. Thus, the probability that the vector that belongs to both $V_1$ and $V_2$ is in $W_1$ is $|W_1| \cdot 2^{-n/2} < 2^{n/10} \cdot  2^{-n/2} = 2^{-2n/5}$.

Finally, the probability that $W_1 \cap W_2 = \emptyset$ is at least the probability that $W_1 \cap V_2 = \emptyset$, and the latter is at least $1-(0.01 + 2^{-2n/5}) \geq 0.9$ for every sufficiently large $n$.
\end{proof}
\end{proof}

\subsection*{Acknowledgments}
We thank Adi Akavia, Shachar Lovett, and Eric Price for useful discussions and comments.

\bibliographystyle{abbrv}
\bibliography{booleanity}

\end{document}